\documentclass[envcountsame]{llncs}

\usepackage{amsfonts,epsfig}
\usepackage{enumerate}
\usepackage{comment}

\begin{document}

\title{Steiner Forest Orientation Problems \thanks{A preliminary version of this article was presented at the 20th European Symposium on Algorithms (ESA 2012).}}

\author{Marek Cygan\inst{1}
 \thanks{Partially supported by National Science Centre grant no. N206 567140, Foundation for Polish Science, ERC Starting Grant NEWNET 279352, NSF CAREER award 1053605, DARPA/AFRL award FA8650-11-1-7162 and ONR YIP grant no. N000141110662.}
\and 
Guy Kortsarz\inst{2} \thanks{Partially supported by NSF support grant award number 0829959.}
\and 
Zeev Nutov\inst{3}
        }

\institute{
  IDSIA, University of Lugano, Switzerland \ \email{marek@idsia.ch}
\and
Rutgers University, Camden \                             \email{guyk@camden.rutgers.edu}
\and
The Open University of Israel \                          \email{nutov@openu.ac.il}
           }

\maketitle

\newcommand {\ignore} [1] {}

\newcommand{\cur}{{\rm cur}}

\begin{abstract}
We consider connectivity problems with orientation constra\-ints.
Given a directed graph $D$ and a collection of ordered node pairs $P$
let $P[D]=\{(u,v) \in P: D \mbox{ contains a } uv\mbox{-path}\}$.
In the {\sf Steiner Forest Orientation} problem we are given  
an undirected graph $G=(V,E)$ with edge-costs and a set $P \subseteq V \times V$ of ordered node pairs. 
The goal is to find a minimum-cost subgraph $H$ of $G$ and an orientation $D$ of $H$ 
such that $P[D]=P$. We give a $4$-approximation algorithm for this problem.

In the {\sf Maximum Pairs Orientation} problem we are given a  
graph $G$ and a multi-collection of ordered node pairs $P$ on $V$.
The goal is to find an orientation $D$ of $G$ such that $|P[D]|$ is maximum.
Generalizing the result of Arkin and Hassin~[DAM'02] for $|P|=2$, we will show that 
for a mixed graph $G$ (that may have both directed and undirected edges),
one can decide in $n^{O(|P|)}$ time whether $G$ has an orientation $D$ with $P[D]=P$
(for undirected graphs this problem admits a polynomial time algorithm for any $P$, but it is NP-complete on mixed graphs).
For undirected graphs, we will show that one can decide 
whether $G$ admits an orientation $D$ with $|P[D]| \geq k$ 
in $O(n+m)+2^{O(k\cdot \log \log k)}$ time; hence this decision problem is 
fixed-parameter tractable, which answers an open question from Dorn et al.~[AMB'11].
We also show that {\sf Maximum Pairs Orientation} admits ratio $O(\log |P|/\log\log |P|)$,
which is better than the ratio $O(\log n/\log\log n)$ of Gamzu et al.~[WABI'10] when $|P|<n$.

Finally, we show that the following node-connectivity problem can be solved in polynomial time:
given a graph $G=(V,E)$ with edge-costs, $s,t \in V$, and an integer $\ell$, 
find a min-cost subgraph $H$ of $G$ with an orientation $D$ 
such that $D$ contains $\ell$ internally-disjoint $st$-paths, 
and $\ell$ internally-disjoint $ts$-paths.
\end{abstract}


\section{Introduction}

\subsection{Problems considered and our results}

We consider connectivity problems with orientation constraints.
Unless stated otherwise, graphs are assumed to be undirected (and may not be simple), 
but we also consider directed graphs,
and even {\em mixed graphs}, which may have both directed and undirected edges.
Given a mixed graph $H$, an {\em orientation of $H$} is a directed graph $D$ 
obtained from $H$ by assigning to each undirected edge one of the two possible directions.
For a mixed graph $H$ on node set $V$ and a multi-collection of 
ordered node pairs (that is convenient to consider as a set of directed edges) $P$ on $V$ 
let $P[H]$ denote the subset of the pairs (or edges) in $P$
for which $H$ contains a $uv$-path. We say that $H$ {\em satisfies} $P$ if $P[H]=P$,
and that $H$ is {\em $P$-orientable} if $H$ admits an orientation $D$ that satisfies $P$.
We note that for undirected graphs 
it is easy to check in polynomial time whether $H$ is $P$-orientable, cf. \cite{HM}
and Section~\ref{s:mix} in this paper.
Let $n=|V|$ denote the number of nodes in $H$ and $m=|E(H)|+|P|$ 
the total number of edges and arcs in $H$ and ordered pairs in $P$.

Our first problem is the classic {\sf Steiner Forest} problem with orientation constraints.

\vspace{0.1cm}

\begin{center} 
\fbox{
\begin{minipage}{0.96\textwidth}
{\sf Steiner Forest Orientation} \\
{\em Instance:} 
A graph $G=(V,E)$ with edge-costs and a set $P \subseteq V \times V$ of 
\hphantom{\em Instance:} \  
ordered node pairs. \\
{\em Objective:}
Find a minimum-cost subgraph $H$ of $G$ with an orientation $D$ that 
\hphantom{\em Objective:} satisfies $P$.
\end{minipage}
}
\end{center}

\begin{theorem} \label{t:min}
{\sf Steiner Forest Orientation} admits a $4$-approximation algorithm. 
\end{theorem}

Our next bunch of results deals with maximization problems of finding an orientation that 
satisfies the maximum number of pairs in $P$. 

\vspace{0.1cm}

\begin{center} 
\fbox{
\begin{minipage}{0.96\textwidth}
{\sf Maximum Pairs Orientation} \\
{\em Instance:} \ 
A graph $G$ and a multi-collection of ordered node pairs (i.e., a set 
\hphantom{\em Instance:} \ of directed edges) $P$ on~$V$. \\
{\em Objective:}
Find an orientation $D$ of $G$ such that the number $|P[D]|$ of pairs 
\hphantom{\em Objective:} satisfied by $D$ is maximum.
\end{minipage}
}
\end{center}

\vspace{0.1cm}

Let {\sf $k$ Pairs Orientation} be the decision problem of determining whether 
{\sf Ma\-xi\-mum Pairs Orientation} has a solution of value at least $k$.
Let {\sf $P$-Orientation} be the decision problem of determining whether 
$G$ is $P$-orientable (this is the {\sf $k$ Pairs Orientation} with $k=|P|$).
As was mentioned, for undirected graphs {\sf $P$-Orientation} can be easily decided 
in polynomial time \cite{HM}.
Arkin and Hassin~\cite{H1} proved that on mixed graphs, {\sf $P$-Orientation}
is NP-complete, but it is polynomial-time solvable for $|P|=2$.
Using new techniques, we widely generalize the result of \cite{H1} as follows.

\begin{theorem} \label{t:mix}
Given a mixed graph $G$ and $P \subseteq V \times V$ one can decide in $n^{O(|P|)}$ time
whether $G$ is $P$-orientable;
namely, {\sf $P$-Orientation} with a mixed graph $G$ can be decided in $n^{O(|P|)}$ time.
In particular, the problem can be decided in polynomial time for any instance with constant $|P|$.
\end{theorem}

In several papers, for example in \cite{ZV}, 
it is stated that any instance of {\sf Ma\-xi\-mum Pairs Orientation} 
admits a solution $D$ such that $|P[D]| \geq |P|/(4\log n)$. 
Furthermore Gamzu et al.~\cite{segev-approx} show that
{\sf Maximum Pairs Orientation} admits an $O(\log n/\log\log n)$-approximation algorithm.
In \cite{fpt} it is shown that {\sf $k$ Pairs Orientation} is fixed-parameter 
tractable\footnote{``Fixed-parameter tractable'' means the following. 
In the parameterized complexity setting, an instance of a decision problem 
comes with an integer parameter $k$.
A problem is said to be {\em fixed-parameter tractable} (w.r.t. $k$) if there exists an algorithm 
that decides any instance $(I,k)$ in time $f(k) {\rm poly}(|I|)$ for some (usually exponential) 
computable function $f$.}
when parameterized by the maximum number of pairs that can be connected via one node.
They posed an open question if the problem is fixed-parameter tractable when parameterized by $k$
(the number of pairs that should be connected),
namely, whether {\sf $k$ Pairs Orientation} can be decided in $f(k) {\rm poly}(n)$ time, 
for some computable function $f$.
Our next result answers this open question, and for $|P|<n$ improves the approximation 
ratio $O(\log n/\log\log n)$ for {\sf Maximum Pairs Orientation} of \cite{ZV,segev-approx}. 

\begin{theorem} \label{t:kernel}
Any instance of {\sf Maximum Pairs Orientation} 
admits a solution $D$, that can be computed in polynomial time, such that $|P[D]| \geq |P|/(4\log_2 (3|P|))$.
Furthermore
\begin{itemize}
\item[{\em (i)}]
{\sf $k$ Pairs Orientation} can be decided in $O(n+m)+2^{O(k\cdot \log \log k)}$ time;
thus it is fixed-parameter tractable when parameterized by $k$.
\item[{\em (ii)}]
{\sf Maximum Pairs Orientation} admits an $O(\log |P|/\log\log |P|)$-approximation algorithm.
\end{itemize}
\end{theorem} 

Note that $|P|$ may be much smaller than $n$, say $|P|=2^{\sqrt{\log n}}$.
While this size of $P$ does not allow exhaustive search in time 
polynomial in $n$, we do get an approximation ratio of 
$O(\sqrt{\log n}/\log\log n)$, 
which is better than the ratio $O(\log n/\log\log n)$ of Gamzu et al.~\cite{segev-approx}.

One may also consider ``high-connectivity'' orientation problems,
to satisfy prescribed connectivity demands.
Several papers considered min-cost edge-connectivity orientation problems, cf. \cite{KNS}.
Almost nothing is known about min-cost node-connectivity orientation problems.
We consider the following simple but still nontrivial variant.

\vspace{0.1cm}

\begin{center} 
\fbox{
\begin{minipage}{0.96\textwidth}
{\sf $\ell$ Disjoint Paths Orientation} \\
{\em Instance:} \  
A graph $G=(V,E)$ with edge-costs, $s,t \in V$, and an integer $\ell$. \\
{\em Objective:}
Find a minimum-cost subgraph $H$ of $G$ with an orientation~$D$ 
\hphantom{\em Objective:} such that $D$ contains $\ell$ internally-disjoint $st$-paths, 
and $\ell$ internally-\hphantom{\em Objective:} disjoint  $ts$-paths.
\end{minipage}
}
\end{center}

\vspace{0.1cm}

Checking whether {\sf $\ell$ Disjoint Paths Orientation} admits a feasible solution can 
be done in polynomial time using the characterization of feasible solutions of 
Egawa, Kaneko, and Matsumoto \cite{EKM} (see Theorem~\ref{t:EKM} in Section~\ref{s:nc});
we use this characterization to prove the following.

\begin{theorem} \label{t:nc}
{\sf $\ell$ Disjoint Paths Orientation} can be solved in polynomial time.
\end{theorem}

Theorems \ref{t:min}, \ref{t:mix}, \ref{t:kernel} and \ref{t:nc}, 
are proved in
sections \ref{s:min}, \ref{s:mix}, \ref{s:kernel} and \ref{s:nc}, 
respectively.

\subsection{Previous and related work}

Let $\lambda_H(u,v)$ denote the $(u,v)$-edge-connectivity in a graph $H$,
namely, the ma\-xi\-mum number of pairwise edge-disjoint $uv$-paths in $H$.
Similarly, let $\kappa_H(u,v)$ denote the $(u,v)$-node-connectivity in $H$, 
namely, the maximum number of pairwise internally node-disjoint $uv$-paths in $H$.
Given an edge-connectivity demand function $r=\{r(u,v):(u,v) \in V \times V\}$, we say that 
$H$ {\em satisfies} $r$ if $\lambda_H(u,v) \geq	r(u,v)$ for all $(u,v) \in V \times V$; similarly, 
for node connectivity demands, we say that $H$ {\em satisfies} $r$ 
if $\kappa_H(u,v) \geq	r(u,v)$ for all $(u,v) \in V \times V$.

\vspace{0.1cm}

\begin{center} 
\fbox{
\begin{minipage}{0.96\textwidth}
{\sf Survivable Network Orientation} \\
{\em Instance:} \ 
A graph $G=(V,E)$ with edge-costs and edge/node-connectivity 
\hphantom{\em Instance:} \ demand function $r=\{r(u,v):(u,v) \in V \times V\}$. \\
{\em Objective:}
Find a minimum-cost subgraph $H$ of $G$ with orientation $D$ that 
\hphantom{\em Objective:} satisfies $r$.
\end{minipage}
}
\end{center}

\vspace{0.1cm}

So far we assumed that the orienting costs are symmetric; 
this means that orienting an undirected edge connecting $u$ and $v$ 
in each one of the two directions is the same, namely, that $c(u,v)=c(v,u)$.
This assumption is reasonable in practical problems,  
but in a theoretic more general setting, we might have non-symmetric costs $c(u,v) \neq c(v,u)$.
Note that the version with non-symmetric costs includes the min-cost version 
of the corresponding directed connectivity problem, and also the case when the input graph 
$G$ is a mixed graph, by assigning large/infinite costs to non-feasible orientations.
For example, {\sf Steiner Forest Orientation} with non-symmetric costs 
includes the {\sf Directed Steiner Forest} problem,
which is {\sf Label-Cover} hard to approximate \cite{DK}.
This is another reason to consider the symmetric costs version.

Khanna, Naor, and Shepherd \cite{KNS} considered several orientation problems with non-symmetric costs.
They showed that when $D$ is required to be $k$-edge-outconnected from a given roots $s$
(namely, $D$ contains $k$ edge-disjoint paths from $s$ to every other node),
then the problem admits a polynomial time algorithm.
In fact they considered a more general problem of finding an orientation that covers 
an intersecting supermodular or crossing supermodular set-function.
See \cite{KNS} for precise definitions. Further generalization of this result due to 
Frank, T.~Kir\'{a}ly, and Z.~Kir\'{a}ly was presented in \cite{FKK}.
For the case when $D$ should be strongly connected, \cite{KNS} obtained 
a $4$-approximation algorithm; note that our {\sf Steiner Forest Orientation} problem 
has much more general demands, that are {\em not} captured by  
intersecting supermodular or crossing supermodular set-functions, but we 
consider symmetric edge-costs (otherwise the problem includes the {\sf Directed Steiner Forest} problem). 
For the case when $D$ is required to be 
$k$-edge-connected, $k \geq 2$, \cite{KNS} obtained a pseudo-approximation algorithm 
that computes a $(k-1)$-edge-connected subgraph of cost at most $2k$ times the cost 
of an optimal $k$-connected subgraph.

We refer the reader to \cite{FK} for a survey on characterization of graphs that admit orientations satisfying 
prescribed connectivity demands, and here mention only 
the following central theorem,
that can be used to obtain a pseudo-approximation for edge-connectivity orientation problems.

\begin{theorem} [Well-Balanced Orientation Theorem, Nash-Williams \cite{NW}] \label{t:NW} \ 
Any undirected graph $H=(V,E_H)$ has an orientation $D$ for which
$\lambda_D(u,v) \geq \left\lfloor \frac{1}{2} \lambda_H(u,v) \right\rfloor$
for all $(u,v) \in V \times V$. 
\end{theorem}

We note that given $H$, an orientation as in Theorem~\ref{t:NW} can be computed in polynomial time.
It is easy to see that if $H$ has 
an orientation $D$ that satisfies $r$ then $H$ satisfies the demand function $q$ defined by 
$q(u,v)=r(u,v)+r(v,u)$.
Theorem~\ref{t:NW} implies that edge-connectivity {\sf Survivable Network Orientation} admits 
a polynomial time algorithm that computes a subgraph $H$ of $G$ and an orientation $D$ of $H$
such that $c(H) \leq 2 {\sf opt}$ and  
$$\lambda_D(u,v) \geq \left \lfloor (r(u,v)+r(v,u))/2 \right\rfloor \geq 
\left\lfloor \max\{r(u,v),r(v,u)\}/2 \right\rfloor \ \ \forall (u,v) \in V \times V \ .$$
This is achieved by applying Jain's \cite{Jain} algorithm to compute a 
$2$-approximate solution $H$ for the corresponding undirected 
edge-connectivity {\sf Survivable Network} instance with demands $q(u,v)=r(u,v)+r(v,u)$,
and then computing an orientation $D$ of $H$ as in Theorem~\ref{t:NW}.
This implies that if the costs are symmetric, then by cost at most $2{\sf opt}$
we can satisfy almost half of the demand of every pair, 
and if also the demands are symmetric then we can satisfy all the demands.
The above algorithm also applies for non-symmetric edge-costs, 
invoking an additional cost factor of $\max_{uv \in E} c(v,u)/c(u,v)$.
Summarizing, we have the following observation, which we failed to find in the literature. 

\begin{corollary} \label{c:NW}
Edge-connectivity {\sf Survivable Network Orientation} (with non-sym\-me\-tric costs) admits 
a polynomial time algorithm that computes a subgraph $H$ of $G$ and an orientation $D$ of $H$
such that $c(H) \leq 2 {\sf opt} \cdot \max_{uv \in E} c(v,u)/c(u,v)$ and  
$\lambda_D(u,v) \geq \left \lfloor \frac{1}{2}(r(u,v)+r(v,u))\right\rfloor$ for all $(u,v) \in V \times V$. 
In particular, the problem  admits a $2$-approximation algorithm if both 
the costs and the demands are symmetric.
\end{corollary}

\section{Algorithm for  {\sf Steiner Forest Orientation} (Theorem~\ref{t:min})}

\label{s:min}

In this section we prove Theorem \ref{t:min}. 
For a mixed graph or an edge set $H$ on a node set $V$ and $X,Y \subseteq V$ let
$\delta_H(X,Y)$ denote the set of all (directed and undirected) edges in $H$ from $X$ to $Y$
and let $d_H(X,Y)=|\delta_H(X,Y)|$ denote their number;
for brevity, $\delta_H(X)=\delta_H(X,\bar{X})$ and $d_H(X)=d_H(X,\bar{X})$,
where $\bar{X}=V \setminus X$.

Given an integral set-function $f$ on subsets of $V$ we say that 
$H$ covers $f$ if $d_H(X) \geq f(X)$ for all $X \subseteq V$.
Define a set-function $f_r$ by $f_r(\emptyset)=f_r(V)=0$ and for every $\emptyset \neq X \subset V$
\begin{equation} \label{e:fr}
f_r(X)=\max \{r(u,v):u \in X, v \in \bar{X}\} + \max \{r(v,u):u \in X, v \in \bar{X}\} \ .
\end{equation}
Note that the set-function $f_r$ is symmetric, namely, that $f_r(X)=f_r(\bar{X})$ for all $X \subseteq V$.

\begin{lemma} \label{l:feasible}
If an undirected edge set $H$ has an orientation $D$ that satisfies an edge-connectivity demand function $r$ then $H$ covers $f_r$. 
\end{lemma}
\begin{proof}
Let $X \subseteq V$. By Menger's Theorem, any orientation $D$ of $H$ that satisfies $r$
has at least $\max \{r(u,v):u \in X, v \in \bar{X}\}$ edges from $X$ to $\bar{X}$,
and at least $\max \{r(v,u):u \in X, v \in \bar{X}\}$ edges from $\bar{X}$ to $X$.
The statement follows.
\qed
\end{proof}

Recall that in the {\sf Steiner Forest Orientation} problem 
we have $r(u,v)=1$ if $(u,v) \in P$ and $r(u,v)=0$ otherwise. 
We will show that if $r_{\max}=\max\limits_{u,v \in V}r(u,v)=1$
then the inverse to Lemma~\ref{l:feasible} is also true, namely,
if $H$ covers $f_r$ then $H$ has an orientation that satisfies $r$;
for this case, we also give a $4$-approximation algorithm for the problem of computing 
a minimum-cost subgraph that covers $f_r$.
We do not know if these results can be extended for $r_{\max} \geq 2$.

\begin{lemma} \label{l:1}
For $r_{\max}=1$, if an undirected edge set $H$ covers $f_r$ then $H$ has an orientation that satisfies $r$.
\end{lemma}
\begin{proof}
Observe that if $(u,v) \in P$ (namely, if $r(u,v)=1$) 
then $u,v$ belong to the same connected component of $H$.
Hence it sufficient to consider the case when $H$ is connected.
Let $D$ be an orientation of $H$ obtained as follows.
Orient every $2$-edge-connected component of $H$ to be strongly connected
(recall that a directed graph is strongly connected if there is a directed path from 
any of its nodes to any other); this is possible by Theorem~\ref{t:NW}.
Now we orient the bridges of $H$. 
Consider a bridge $e$ of $H$. 
The removal of $e$ partitions $V$ into two connected
components $X,\bar{X}$.
Note that $\delta_P(X,\bar{X})=\emptyset$ or 
$\delta_P(\bar{X},X)=\emptyset$, since $f_r(X) \leq d_H(X)=1$.
If $\delta_P(X,\bar{X}) \neq \emptyset$, we orient $e$ from $X$ to $\bar{X}$;
if $\delta_P(\bar{X},X) \neq \emptyset$, we orient $e$ from $\bar{X}$ to $X$;
and if $\delta_P(X,\bar{X}),\delta_P(\bar{X},X)=\emptyset$, we orient $e$ arbitrarily.
It is easy to see that the obtained orientation $D$ of $H$ satisfies $P$.
\qed
\end{proof}

We say that an edge-set or a graph $H$ {\em covers} a set-family ${\cal F}$ if 
$d_H(X) \geq 1$ for all $X \in {\cal F}$.
A set-family ${\cal F}$ is said to be {\em uncrossable} if for any $X,Y \in {\cal F}$ 
the following holds: $X \cap Y, X \cup Y \in {\cal F}$ or 
$X \setminus Y,Y \setminus X \in {\cal F}$.
The problem of finding a minimum-cost set of undirected edges that covers an 
uncrossable set-family ${\cal F}$ admits a primal-dual $2$-approximation algorithm, 
provided the inclusion-minimal
members of ${\cal F}$ can be computed in polynomial time \cite{GGPS}.
It is known that the undirected {\sf Steiner Forest} problem is a particular case of the 
problem of finding a min-cost cover of an uncrossable family, and thus admits a 
$2$-approximation algorithm.
 
\begin{lemma} \label{l:F-uncross}
Let $H=(V,J \cup P)$ be a mixed graph, where edges in $J$ are undirected and edges in $P$ are directed, such that for every $uv \in P$ both 
$u,v$ belong to the same connected component of the graph $(V,J)$. Then the set-family 
${\cal F}=\{S \subseteq V: d_J(S)=1 \wedge d_P(S),d_P(\bar{S}) \geq 1\}$ is 
uncrossable, and its inclusion minimal members can be computed in polynomial time.
\end{lemma} 
\begin{proof}
Let ${\cal C}$ be the set of connected components of the graph $(V,J)$.
Let $C \in {\cal C}$.
Any bridge $e$ of $C$ partitions $C$ into two parts $C'(e),C''(e)$
such that $e$ is the unique edge in $J$ connecting them. 
Note that the condition $d_J(S)=1$ is equivalent to the following condition (C1), 
while if condition (C1) holds then the condition $d_P(S),d_P(\bar{S}) \geq 1$ 
is equivalent to the following condition (C2)
(since no edge in $P$ connects two distinct connected components of $(V,J)$).
\begin{enumerate}[(C1)]
\item
There exists $C_S \in {\cal C}$ and a bridge $e_S$ of $C_S$,
such that $S$ is a union of one of the sets $X'=C'_S(e_S),X''=C''_S(e_S)$
and sets in ${\cal C} \setminus \{C_S\}$.
\item
$d_P(X',X''),d_P(X'',X') \geq 1$.
\end{enumerate}
Hence we have the following characterization of the sets in ${\cal F}$:
{\em $S \in {\cal F}$ if, and only if, conditions (C1), (C2) hold for $S$.} 
This implies that every inclusion-minimal member of ${\cal F}$ is 
$C'(e)$ or $C''(e)$, for some bridge $e$ of $C \in {\cal C}$.
In particular, the inclusion-minimal members of ${\cal F}$ 
can be computed in polynomial time.


\begin{figure} 
\centering 
\epsfbox{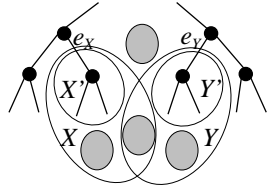}
   \caption{Illustration to the proof of Lemma~\ref{l:F-uncross};
            components distinct from $C_X,C_Y$ are shown by gray ellipses.} 
\label{f:comps}
\end{figure}

Now let $X,Y \in {\cal F}$ (so conditions (C1), (C2) hold for each one of $X,Y$), 
let $C_X,C_Y \in {\cal C}$ be the corresponding connected components and 
$e_X,e_Y$ the corresponding bridges 
(possibly $C_X=C_Y$, in which case we also may have $e_X=e_Y$), and let 
$X',X''$ and $Y',Y''$ be the corresponding partitions of $C_X$ and $C_Y$, respectively.
Since $e_X,e_Y$ are bridges, at least one of the sets 
$X' \cap Y',X' \cap Y'',X'' \cap Y',X'' \cap Y''$
must be empty, say $X' \cap Y'=\emptyset$.
Note that the set-family ${\cal F}$ is symmetric, hence to prove that 
$X \cap Y,X \cup Y \in {\cal F}$ or $X \setminus Y,Y \setminus X \in {\cal F}$,
it is sufficient to prove 
that $A \setminus B, B \setminus A \in {\cal F}$ 
for some pair $A,B$ such that $A \in \{X,\bar{X}\},B \in \{Y,\bar{Y}\}$.
E.g., if $A=X$ and $B=\bar{Y}$, then 
$A \setminus B=X \cap Y$ and $B \setminus A=V \setminus (X \cup Y)$, hence 
$A \setminus B, B \setminus A \in {\cal F}$ together with the 
symmetry of ${\cal F}$ implies $X \cap Y, X \cup Y \in {\cal F}$.
Similarly, if $A=\bar{X}$ and $B=\bar{Y}$, then 
$A \setminus B=Y \setminus X$ and 
$B \setminus A=X \setminus Y$, hence 
$A \setminus B, B \setminus A \in {\cal F}$  
implies $Y \setminus X, X \setminus Y \in {\cal F}$.
Thus w.l.o.g. we may assume that 
$X' \subseteq X$ and $Y' \subseteq Y$, see Figure~\ref{f:comps},
and we show that $X \setminus Y,Y \setminus X \in {\cal F}$.
Recall that $X' \cap Y' = \emptyset$ and
hence $X \cap Y$ is a (possibly empty) union of some sets in 
${\cal C} \setminus \{C_X,C_Y\}$.
Thus $X \setminus Y$ is a union of $X'$ and some sets in 
${\cal C} \setminus \{C_X,C_Y\}$.
This implies that conditions (C1), (C2) hold for $X \setminus Y$, hence 
$X \setminus Y \in {\cal F}$; the proof that
$Y \setminus X \in {\cal F}$ is similar.
This concludes the proof of the lemma.
\qed
\end{proof}

\begin{lemma} \label{l:3}
Given a {\sf Steiner Forest Orientation} instance, the problem of compu\-ting a 
minimum-cost subgraph $H$ of $G$ that covers $f_r$ admits a $4$-approximation algorithm.
\end{lemma}
\begin{proof}
The algorithm has two phases.
In the first phase we solve the corresponding undirected {\sf Steiner Forest} instance
with the same demand function $r$. The {\sf Steiner Forest} 
problem admits a $2$-approximation algorithm, hence $c(J) \leq 2 {\sf opt}$.
Let $J$ be a subgraph of $G$ computed by such a $2$-approximation algorithm.
Note that $f_r(S)-d_J(S) \leq 1$ for all $S \subseteq V$.
Hence to obtain a cover of $f_r$ it is sufficient to cover 
the family ${\cal F}=\{S \subseteq V:f_r(S)-d_J(S)=1\}$ of the deficient sets w.r.t. $J$.
The key point is that the family ${\cal F}$ is uncrossable, and that the inclusion-minimal 
members of ${\cal F}$ can be computed in polynomial time.
In the second phase we compute a $2$-approximate cover of this ${\cal F}$ using the 
algorithm of \cite{GGPS}.
Observe that the set $E(H) \setminus E(J)$, that is 
the set of edges of the optimum solution with edges of $J$ removed,
covers the family ${\cal F}$ and therefore the cost of the second phase 
is at most $2 {\sf opt}$.
Consequently, the problem of covering $f_r$ is reduced to solving two problems
of covering an uncrossable set-family.  

To show that ${\cal F}$ is uncrossable we use Lemma~\ref{l:F-uncross}.
Note that for any $(u,v) \in P$ both $u,v$ belong to the same connected component of $(V,J)$, 
and that $f_r(S)-d_J(S)=1$ if, and only if, $d_J(S)=1$ and $d_P(S),d_P(\bar{S}) \geq 1$,
hence ${\cal F}=\{S \subseteq V: d_J(S)=1 \wedge d_P(S),d_P(\bar{S}) \geq 1\}$.
Consequently, by Lemma~\ref{l:F-uncross}, the family ${\cal F}$ is uncrossable and its 
inclusion-minimal members can be computed in polynomial time.
This concludes the proof of the lemma.
\qed
\end{proof}

The proof of Theorem~\ref{t:min} is complete.

\section{Algorithm for {\sf $P$-Orientation} on mixed graphs (Theorem~\ref{t:mix})} \label{s:mix}

In this section we prove Theorem \ref{t:mix}.
The following (essentially known) statement is straightforward. 

\begin{lemma} \label{lem:contract-cycle}
Let $G$ be a mixed graph, let $P$ be a set of directed edges on $V$, 
and let $C$ be a subgraph of $G$ that admits a strongly connected orientation. 
Let $G',P'$ be obtained from $G,P$ by contracting $C$ into a single node.
Then $G$ is $P$-orientable if, and only if, $G'$ is $P'$-orientable. 
In particular, this is so if $C$ is a cycle. \hfill $\Box$
\end{lemma}

\begin{corollary} \label{cor:tree}
{\sf $P$-orientation} (with an undirected graph $G$) can be decided in polynomial time.
\end{corollary}

\begin{proof}
By repeatedly contracting a cycle of $G$, we obtain an equivalent instance,
by Lemma~\ref{lem:contract-cycle}.
Hence we may assume that $G$ is a tree.
Then for every $(u,v)\in P$ there is a unique $uv$-path in $G$, 
which imposes an orientation on all the edges of this path.
Hence if suffices to check that no two pairs in $P$ 
impose different orientations of the same edge of the tree.
\qed
\end{proof}

Our algorithm for mixed graphs is based on a similar idea.
We say that a mixed graph is an {\em ori-cycle} if it admits an orientation that is a 
directed simple cycle. We need the following statement.

\begin{lemma} \label{lem:find-cycle}
Let $G=(V,E \cup A)$ be a mixed graph, where edges of $E$ are undirected and $A$ contains directed arcs, and let $G'$ be obtained from $G$ by contracting every
connected component of the undirected graph $(V,E)$ into a single node.
If there is a directed cycle (possibly  a self-loop) $C'$ in $G'$ then there is an ori-cycle $C$ in $G$,
and such $C$ can be found in polynomial time.
\end{lemma}
\begin{proof}
If $C'$ is also a directed cycle in $G$, then we take $C=C'$.
Otherwise, we replace every node $v_X$ of $C'$ that corresponds to a contracted connected component $X$ of $(V,E)$
by a path, as follows.
Let $a_1$ be the arc entering $v_X$ in $C'$ and let $a_2$ be the arc leaving $v_X$ in $C'$.
Let $v_1$ be the head of $a_1$ and similarly let $v_2$ be a the tail of $a_2$.
Since $X$ is a connected component in $(V,E)$, there is a $v_1v_2$-path in $(V,E)$, and we replace $X$ by 
this path. The result is the required ori-cycle $C$ (possibly a self-loop) in $G$.
It is easy to see that such $C$ can be obtained from $C'$ in polynomial time. 
\qed
\end{proof}

By Lemmas~\ref{lem:find-cycle} and \ref{lem:contract-cycle}
we may assume that the directed graph $G'$ obtained from $G$ by contracting every
connected component of $(V,E)$, is a directed acyclic multigraph (with no self-loops).
This preprocessing step is similar to the one used by Silverbush et al.~\cite{silverbush}.
Let $p=|P|$. Let $f:V \rightarrow V(G')$ be the function
which for each node $v$ of $G$ assigns a node $f(v)$ in $G'$
that represents the connected component of $(V,E)$ that contains $v$
(in other words the function $f$ shows a correspondence between
nodes before and after contractions).

The first step of our algorithm is to guess the first and the last edge
on the path for each of the $p$ pairs in $P$, by trying all $n^{O(p)}$ possibilities.
If for the $i$-th pair an undirected edge is selected as the first or the last one
on the corresponding path, then we orient it accordingly and move it from $E$ to $A$.
Thus by functions ${\rm last},{\rm first}\ :\ \{1,\ldots,p\} \rightarrow A$ we denote
the guessed first and last arc for each of the $p$ paths.

Now we present a branching algorithm with exponential time complexity
which we later convert to $n^{O(p)}$ time by applying a method of memoization.
Let $\pi$ be a topological ordering of $G'$.
By $\cur:\{1,\ldots,p\} \rightarrow A$ we denote the most recently chosen
arc from $A$ for each of the $p$ paths (initially $\cur(i)={\rm first}(i)$).
In what follows we consider subsequent nodes $v_C$ of $G'$ with respect to $\pi$
and branch on possible orientations of the connected component $C$ of $G$.
We use this orientation to update the function $\cur$ for all the arguments $i$
such that $\cur(i)$ is an arc entering a node mapped to $v_C$.

Let $v_C \in V(G')$ be the first node w.r.t. to $\pi$
which was not yet considered by the branching algorithm.
Let $I \subseteq \{1,\ldots,p\}$ be the set of indices $i$ such that $\cur(i)=(u,v) \in A$ for $f(v)=v_C$,
and $\cur(i) \not={\rm last}(i)$.
If $I=\emptyset$ then we skip $v_C$ and proceed to the next node in $\pi$.
Otherwise for each $i \in I$ we branch on choosing an arc $(u,v) \in A$ such that $f(u)=v_C$,
that is we select an arc that the $i$-path will use just after leaving the connected component of $G$
corresponding to the node $v_C$ (note that there are at most $|A|^{|I|}=n^{O(p)}$ branches).
Before updating the arcs $\cur(i)$ for each $i\in I$ in a branch, we check whether the connected component $C$
of $(V,E)$ consisting of nodes $f^{-1}(v_C)$ is accordingly orientable by using Corollary~\ref{cor:tree} (see Figure~\ref{fig:zoom}).
Finally after considering all the nodes in $\pi$ we check whether for each $i \in \{1,\ldots,p\}$ 
we have $\cur(i) = {\rm last}(i)$.
If this is the case our algorithm returns YES and otherwise it returns NO in this branch.

\begin{figure} 
\centering 
\epsfysize=3.1cm
\epsfbox{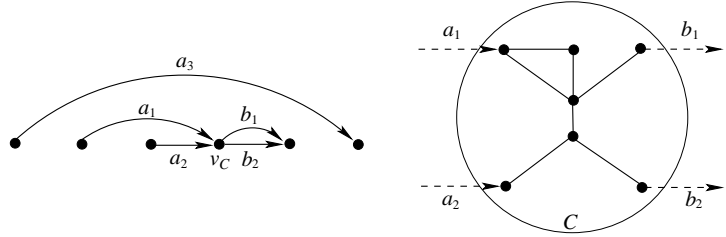}
   \caption{Our algorithm considers what orientation the connected component $C$ (of $(V,E)$)
will have. Currently we have $\cur(1)=a_1$, $\cur(2)=a_2$ and $\cur(3)=a_3$, hence $I=\{1,2\}$.
If in a branch we set new values $\cur(1)=b_1$ and $\cur(2)=b_2$ then by Corollary~\ref{cor:tree} we can
verify that it is possible to orient $C$, so that there is a path from the end-point 
of $a_1$ to the start-point of $b_1$ and from the end-point of $a_2$ to the start-point of $b_2$.
However the branch with new values $\cur(1)=b_2$ and $\cur(2)=b_1$ will be terminated,
since it is not possible to orient $C$ accordingly.}
\label{fig:zoom}
\end{figure}

The correctness of our algorithm follows from the invariant that each node $v_C$ is considered
at most once, since all the updated values $\cur(i)$ are changed to arcs that are
to the right with respect to~$\pi$.

Observe that when considering a node $v_C$ it is not important what orientations previous nodes in $\pi$ have,
because all the relevant information is contained in the $\cur$ function. 
Therefore to improve the currently exponential time complexity we apply the standard technique of 
memoization, that is, store results of all the previously computed recursive calls.
Consequently for any index of the currently considered node in $\pi$ 
and any values of the function $\cur$, there is at most one branch
for which we compute the result, since for the subsequent recursive calls
we use the previously computed results.
This leads to $n^{O(p)}$ branches and $n^{O(p)}$ total time and space complexity.

\section{Algorithms for {\sf Maximum Pairs Orientation} (Theorem~\ref{t:kernel})} \label{s:kernel}

In this section we prove Theorem \ref{t:kernel}.

\begin{lemma}
\label{lem:kernel}
There exists a linear time algorithm that given an instance 
of {\sf Maximum Pairs Orientation} or of {\sf k Pairs Orientation},
transforms it into an equivalent instance such that the input graph is a tree
with at most $3p-1$ nodes.
\end{lemma}

\begin{proof}
As is observed in \cite{HM}, and also follows from Lemma~\ref{lem:contract-cycle},
we can assume that the input graph $G$ is a tree; such a tree can be constructed in linear time
by contracting the $2$-edge-connected components of $G$.
For each edge $e$ of $G$ we compute an integer $P(e)$,
that is the number of pairs $(s,t)$ in $P$, such that $e$
belongs to the shortest path between $s$ and $t$ in $G$.
If for an edge $e$ we have $P(e) \le 1$, we contract this edge,
since we can always orient it as desired by at most one $st$-path.
Note that after this operation each leaf belongs to at least two pairs in $P$.
If a node $v$ has degree $2$ in the tree and does not belong to any pair,
we contract one of the edges incident to $v$; this is since in any inclusion minimal solution,
one of the two edges enters $v$ if, and only if, the other leaves~$v$. 

The linear time implementation of the presented reductions is as follows.
First, using a linear time algorithm for computing 2-edge-connected components,
and by scanning every edge in $E \cup P$, we can see which components it connects,
thus obtaining an equivalent instance where $G$ is a tree.
 To compute all the values $P(e)$, we root the tree in an arbitrary node
and create a multiset of pairs $P'=\{(s,{\rm lca}(s,t)), (t,{\rm lca}(s,t)): (s,t) \in P\}$,
where ${\rm lca}(s,t)$ is the lowest common ancestor of $s$ and $t$,
which can be computed in linear time~\cite{HT}.
Note that in $P'$ the first coordinate of each pair is an descendant of the second coordinate node
and pairs in $P'$ represent upward paths.
Let $a(v) = |\{(v,x) \in P'\}| - |\{(x,v) \in P'\}|$ and observe that 
if $u$ is a parent of $v$ in the tree $G$, then $P(uv)$ equals
the sum of values $a(u')$ over all descendants $u'$ of the node $u$ (including $u'=u$).
Therefore we can compute all the values $P(e)$ in linear time
and contract all the edges with $P(e)=1$.
Note that after the contractions are done we need to relabel pairs in $P$,
since some nodes may have their labels changed,
but this can be also done in linear time by storing a new label for each node
in a table.
Finally using a graph search algorithm we find maximal paths such
that each internal node is of degree two and does not belong to any pair.
For each such path we contract all but one edge.

We claim that after these reductions are implemented, the tree $G'$ obtained has at most $3p-1$ nodes. 
Let $\ell$ be the number of leaves and $t$ the number of nodes of degree $2$ in $G'$. 
As each node of degree less than $3$ in $G'$ is an $s_i$ or $t_i$, $\ell+t\leq 2p$.
Since each leaf belongs to at least two pairs in $P$, $\ell \le p$.
The number of nodes of degree at least $3$ is at most 
$\ell-1$ and so $|V(G')|\leq 2\ell+t-1\leq 3p-1$.
This concludes the proof of the lemma.
\qed
\end{proof}

After applying Lemma~\ref{lem:kernel}, the number of nodes $n$
of the returned tree is at most $3p-1$.
Therefore, by~\cite{ZV}, one can find in polynomial time
a solution $D$, such that $|P[D]| \ge p/(4\log_2n) \ge p/(4\log_2(3p))$.
Therefore, if for a given {\sf $k$ Pairs Orientation} instance
we have $k \le p/(4\log_2(3p))$, then clearly it is a YES instance.
However if $k>p/(4\log_2(3p))$, then $p=\Theta(k\log k)$. 
In order to solve the {\sf $k$ Pairs Orientation} instance
we consider all possible ${p \choose k}$ subsets $P'$ of exactly $k$
pairs from $P$, and check if the graph is $P'$-orientable.
Observe that
$${p \choose k}\le \frac{p^k}{k!} \le \frac{p^k}{(k/e)^k} \le \frac{p^k}{(p/(4e\log_2 (3p)))^k} = (4e\log_2 (3p))^k = 2^{O(k\log\log k)}$$
where the second inequality follows from Stirling's formula.
Therefore the running time is $O(m+n)+2^{O(k\log\log k)}$, which proves (i).

Combining Lemma~\ref{lem:kernel} with the $O(\log n/\log \log n)$-approximation algorithm of 
Gamzu et al. \cite{segev-approx} proves (ii).
Thus the proof of Theorem~\ref{t:kernel} is complete.

\section{Conclusions and open problems}

In this paper we considered minimum-cost and maximum pairs orientation problems.
Our main results are a $4$-approximation algorithm for {\sf Steiner Forest Orientation},
an $n^{O(|P|)}$ time algorithm for {\sf $P$-Orientation} on mixed graphs,
and an $O(n+m)+2^{O(k\cdot \log \log k)}$ time algorithm for {\sf $k$ Pairs Orientation} 
(which implies that {\sf $k$ Pairs Orientation} is fixed-parameter tractable when parameterized by $k$,
solving an open question from \cite{fpt}).
We now mention some open problems, most of them related to the work of Khanna, Naor, and Shepherd \cite{KNS}.

We have shown that the {\sf $P$-Orientation} problem on mixed graphs parameterized by $|P|$ belongs to XP,
however to the best of our knowledge it is not known whether this problem is fixed parameter tractable
or $W[1]$-hard.

As was mentioned, \cite{KNS} 
showed that the problem of computing a minimum-cost 
$k$-edge-outconnected orientation can be solved in polynomial time, even for non-symmetric edge-costs.
To the best of our knowledge, for {\em node-connectivity}, and even for the simpler notion 
of {\em element-connectivity}, no non-trivial approximation ratio is known even for 
symmetric costs and $k=2$.
Moreover, even the decision version of determining whether 
an undirected graph admits a $2$-outconnected orientation is not known to be in P nor NP-complete.  

For the case when the orientation $D$ is required to be 
$k$-edge-connected, $k \geq 2$, \cite{KNS} obtained a pseudo-approximation algorithm 
that computes a $(k-1)$-edge-connected subgraph of cost at most $2k$ times the cost 
of an optimal $k$-connected subgraph.
It is an open question if the problem admits a non-trivial true approximation 
algorithm even for $k=2$.

\section{Algorithm for {\sf $\ell$ Disjoint Paths Orientation} (Theorem~\ref{t:nc})} \label{s:nc}

In this section we prove Theorem~\ref{t:nc}.
We need the following characterization due to \cite{EKM} 
of feasible solutions to {\sf $\ell$ Disjoint Paths Orientation}.

\begin{theorem} [\cite{EKM}] \label{t:EKM}
Let $H=(V,E_H)$ be an undirected graph and let $s,t \in V$. 
Then $H$ has an orientation $D$ such that 
$\kappa_D(s,t),\kappa_D(t,s) \geq \ell$ if, and only if,
\begin{equation} \label{e:EKM}
\lambda_{H \setminus C}(s,t) \geq 2(\ell-|C|) \ \mbox{ for every }  C \subseteq V \setminus \{s,t\} 
\mbox{ with } |C|<\ell \ .
\end{equation}
Furthermore, if $H$ satisfies (\ref{e:EKM}), then an orientation $D$ of $H$
that satisfies \\ 
$\kappa_D(s,t),\kappa_D(t,s) \geq \ell$ can be computed in polynomial time. 
\end{theorem}

Now, let us use the following version of Menger's Theorem for node and edge capacitated graphs;
this version can be deduced from the original Menger's Theorem by elementary constructions.

\begin{lemma} \label{l:M}
Let $s,t$ be two nodes in a directed/undirected graph $H=(V,E)$ with edge and node capacities 
$\{u(a):a \in E \cup (V \setminus \{s,t\}\}$. 
Then the maximum number of $st$-paths such that every $a \in E \cup (V \setminus \{s,t\})$ 
appears in at most $u(a)$ of them equals to
$\min\{u(A): A \subseteq E \cup (V \setminus \{s,t\}), \lambda_{H \setminus A}(s,t)=0\}$. \hfill $\Box$
\end{lemma}

From Lemma~\ref{l:M} we deduce the following.

\begin{corollary} \label{c:EKM}
An undirected graph $H=(V,E)$ satisfies (\ref{e:EKM}) if, and only if, the following condition holds: 
$H$ contains $2\ell$ edge-disjoint $st$-paths such that every $v \in V \setminus \{s,t\}$ 
belongs to at most 2 of them.
\end{corollary}
\begin{proof}
Assign capacity $u(e)=1$ to every $e \in E$ and capacity $u(v)=2$ to every $v \in V \setminus \{s,t\}$.
By Lemma~\ref{l:M}, the condition in the corollary is equivalent to the condition
$$\min\{u(C \cup F): F \subseteq E, C \subseteq (V \setminus \{s,t\}),\lambda_{H \setminus (C \cup F)}(s,t)=0\} 
\geq 2\ell \ .$$
Since $u(C)=2$ for all $C \subseteq V \setminus \{s,t\}$ and $u(F)=|F|$ for all $F \subseteq E$, 
the latter condition is equivalent to the condition
$$\min \{|F|: F \subseteq E,\lambda_{(H \setminus C) \setminus F}(s,t)=0\} \geq 2\ell -2|C| \ \ 
\forall  C \subseteq V \setminus \{s,t\} \mbox{ with } |C|<\ell \ .$$
The above condition is equivalent to (\ref{e:EKM}), since 
for every $C \subseteq V \setminus \{s,t\}$ we have 
$\min \{|F|: F \subseteq E,\lambda_{(H \setminus C) \setminus F}(s,t)=0\}=\lambda_{H \setminus C}(s,t)$,
by applying Menger's Theorem on the graph $H \setminus C$.
The statement follows.
\qed
\end{proof}

Now consider the following problem.

\vspace{0.1cm}

\begin{center} 
\fbox{
\begin{minipage}{0.96\textwidth}
{\sf Node-Capacitated Min-Cost $k$-Flow} \\ 
{\em Instance:} \ A graph $G=(V,E)$ with edge-costs,
$s,t \in V$, node-capacities  
\hphantom{\em Instance:} \ 
$\{b_v: v \in V \setminus \{s,t\}\}$, and an integer $k$. \\
{\em Objective:} Find a set $\Pi$ of $k$ edge-disjoint paths such that every $v \in V \setminus \{s,t\}$ 
\hphantom{\em Objective:} belongs to at most $b_v$ paths in $\Pi$.
\end{minipage}
}
\end{center}

\vspace{0.1cm}

From Corollary~\ref{c:EKM}, we see that {\sf $\ell$ Disjoint Paths Orientation} is a particular case of 
{\sf Node-Capacitated Min-Cost $k$-Flow} when $H$ is undirected, $k=2\ell$, and all node capacities are $2$.

{\sf Node-Capacitated Min-Cost $k$-Flow} can be solved in polynomial time, for both directed and undirected graphs,
by reducing the problem to the standard {\sf Edge-Capacitated Min-Cost $k$-Flow} problem.
For directed graphs this can be done by a standard reduction of converting node-capacities
to edge-capacities: replace every node $v \in V \setminus\{s,t\}$ by the two
nodes $v^+,v^-$, connected by the edge $v^+v^-$ having the same capacity as $v$, and redirect the
heads of the edges entering $v$ to $v^+$ and the tails of the edges leaving $v$ to $v^-$.
The undirected case is easily reduced to the directed one, by solving the problem on the 
bidirection graph of $G$, obtained from $G$ by replacing every undirected edge $e$ connecting $u,v$
by a pair of antiparallel directed edges $uv,vu$ of the same cost as $e$.

The proof of Theorem~\ref{t:nc} is complete.

\bibliographystyle{abbrv}
\bibliography{u}

\begin{thebibliography}{10}

\bibitem{H1}
E.~Arkin and R.~Hassin.
\newblock A note on orientations of mixed graphs.
\newblock {\em Discrete Applied Mathematics}, 116(3):271--278, 2002.

\bibitem{DK}
Y.~Dodis and S.~Khanna.
\newblock Design networks with bounded pairwise distance.
\newblock In {\em STOC}, pages 750--759, 1999.

\bibitem{fpt}
B.~Dorn, F.~H\"{u}ffner, D.~Kr\"{u}ger, R.~Niedermeier, and J.~Uhlmann.
\newblock Exploiting bounded signal flow for graph orientation based on
  cause-effect pairs.
\newblock {\em Algorithms for Molecular Biology}, 6(21), 2011.

\bibitem{EKM}
Y.~Egawa, A.~Kaneko, and M.~Matsumoto.
\newblock A mixed version of {M}enger's theorem.
\newblock {\em Combinatorica}, 11:71--74, 1991.

\bibitem{FK}
A.~Frank and T.~Kir\'{a}ly.
\newblock Combined connectivity augmentation and orientation problems.
\newblock {\em Discrete Applied Mathematics}, 131(2):401--419, 2003.

\bibitem{FKK}
A.~Frank, T.~Kir\'{a}ly, and Z.~Kir\'{a}ly.
\newblock On the orientation of graphs and hypergraphs.
\newblock {\em Discrete Applied Mathematics}, 131:385–400, 2003.

\bibitem{segev-approx}
I.~Gamzu, D.~Segev, and R.~Sharan.
\newblock Improved orientations of physical networks.
\newblock In {\em WABI}, pages 215--225, 2010.

\bibitem{GGPS}
M.~Goemans, A.~Goldberg, S.~Plotkin, D.~Shmoys, E.~Tardos, and D.~Williamson.
\newblock Improved approximation algorithms for network design problems.
\newblock In {\em SODA}, pages 223--232, 1994.

\bibitem{HT}
D.~Harel and R.~E. Tarjan.
\newblock Fast algorithms for finding nearest common ancestors.
\newblock {\em SIAM J. Comput.}, 13(2):338--355, 1984.

\bibitem{HM}
R.~Hassin and N.~Megiddo.
\newblock On orientations and shortest paths.
\newblock {\em Linear Algebra and Its Applications}, pages 589--602, 1989.

\bibitem{Jain}
K.~Jain.
\newblock A factor 2 approximation algorithm for the generalized {Steiner}
  network problem.
\newblock {\em Combinatorica}, 21(1):39--60, 2001.

\bibitem{KNS}
S.~Khanna, J.~Naor, and B.~Shepherd.
\newblock Directed network design with orientation constraints.
\newblock In {\em SODA}, pages 663--671, 2000.

\bibitem{ZV}
A.~Medvedovsky, V.~Bafna, U.~Zwick, and R.~Sharan.
\newblock An algorithm for orienting graphs based on cause-effect pairs and its
  applications to orienting protein networks.
\newblock In {\em WABI}, pages 222--232, 2008.

\bibitem{NW}
Nash-Williams.
\newblock On orientations, connectivity and odd vertex pairings in finite
  graphs.
\newblock {\em Canad. J. Math.}, 12:555--567, 1960.

\bibitem{silverbush}
D.~Silverbush, M.~Elberfeld, and R.~Sharan.
\newblock Optimally orienting physical networks.
\newblock In {\em RECOMB}, pages 424--436, 2011.

\end{thebibliography}

\end{document}